\documentclass[nofootinbib,twocolumn,pra,letterpaper,longbibliography]
{revtex4-2}
\usepackage{hyperref}
\usepackage{latexsym,amsmath,amssymb,amsfonts,graphicx,color,amsthm}
\usepackage{enumerate}
\usepackage{braket}
\usepackage{adjustbox}
\usepackage{footnote}
\usepackage[dvipsnames]{xcolor}
\usepackage{tikz}
\usetikzlibrary{positioning,calc,matrix,shapes.multipart}

\newcommand{\cC}{\mathcal{C}}

\newcommand{\cH}{\mathcal{H}}
\newcommand{\cI}{\mathcal{I}}

\newcommand{\cK}{\mathcal{K}}
\newcommand{\cL}{\mathcal{L}}

\newcommand{\cO}{\mathcal{O}}

\newcommand{\cS}{\mathcal{S}}

\newcommand{\Id}{\mathbb{I}}
\newcommand{\tr}{\text{Tr}}

\newtheorem{theorem}{Theorem}
\newtheorem{proposition}{Proposition}

\newtheorem{corollary}[theorem]{Corollary}
\newtheorem{definition}{Definition}
\newtheorem{example}{Example}

\begin{document}

\title{Layers of classicality in the compatibility of measurements}
\author{Arindam Mitra$^{1,2}$}
\affiliation{$^1$Optics and Quantum Information Group, The Institute of Mathematical Sciences,
C. I. T. Campus, Taramani, Chennai 600113, India.\\
$^2$Homi Bhabha National Institute, Training School Complex, Anushakti Nagar, Mumbai 400094, India.}

\date{\today}

\begin{abstract}
The term ``Layers of classicality" in the context of quantum measurements, was introduced in [T. Heinosaari, Phys. Rev. A \textbf{93}, 042118 (2016)]. The strongest layer among these consists of the sets of observables that can be broadcast and the weakest layer consists of the sets of compatible observables. There are several other layers in between those two layers. In this work, we study their physical and geometric properties and show the differences and similarities among the layers in these properties. In particular we show that: (i) none of the layers of classicality respect transitivity property, (ii) the concept like degree of broadcasting similar to degree of compatibility does not exist,  (iii) there exist informationally incomplete POVMs that are not individually broadcastable, (iv)  a set of broadcasting channels can be obtained through concatenation of broadcasting and non-disturbing channels, (v) unlike compatibility, other layers of classicality are not convex, in general. Finally, we discuss the relations among these layers. More specifically, we show  that specific type of concatenation relations among broadcasting channels decide the layer in which a pair of observables resides.
\end{abstract}

\maketitle
\section{Introduction}
The incompatibility of observables is one of the origins of nonclassicality in quantum theory. A set of POVMs is compatible if those can be measured simultaneously. Otherwise, they are incompatible \cite{heinosaari}. Today we all know the connections between incompatibility, nonlocality and steering \cite{Brunner__nonloc_incomp,bush_bell_steer_incomp}. Nonclassical features like Bell inequality violation as well as the demonstration of steering is  only possible using incompatible measurements \cite{Barnett_nonloc_incomp, Uola_steer_incomp}. In quantum information theory, incompatibility of measurements provide an advantage over compatible measurements in several information-theoretic tasks \cite{mubs,qkd_d}.\\
The impossibility of universal cloning of quantum state is known as the No-cloning theorem\cite{wootters-no_clon}. Broadcasting is the weaker version of cloning. Today we know that only commuting states can be simultaneously broadcast \cite{barnum}. Similarly, there are no-cloning theorems for a single POVM as well as for a set of POVMs \cite{Paris_clon_obser,Paris_clon_obser_joint_measure,Rodriguez_clon,Rastegin}.\\
Compatible measurement does not show any nonclassical feature. But the compatibility of observables is the weakest layer of classicality. The other layers from stronger to weaker can be written as (1) broadcastable sets of observables, (2) one-side broadcastable sets of observables, (3) mutually nondisturbing observables, and (4) nondisturbing observables \cite{heino_layers_of_classicality}.\\

In the Ref.  \cite{heino_layers_of_classicality}, the existence of hierarchies of these layers have been proved and complete characterisation for qubit observables is presented and also the the work presented in Ref. \cite{heino_layers_of_classicality} is the first work in this research direction.
In the present article, we discuss the properties of individual layers and the present work may be considered as the immediate development in this research direction subsequent to the work prsentated in the Ref. \cite{heino_layers_of_classicality}. More specifically in this paper, we mainly tackle the following questions- (i) what are the differences in properties between compatibility and any other layer of classicality? (ii) are there any similarities in the properties between compatibility and any other layer of classicality? (iii) what are the connections among different layers of classicality?  Our results open up several research avenues which has been summarised in the conclusion.\\

The rest of this paper is organized as follows. In  section \ref{sec:prelimi}, we briefly review the incompatibility of quantum observables, quantum channels and the layers of classicality. In section \ref{sec:main}, we discuss our main results i.e., physical as well as geometric properties of different layers of classicality and relations among different layers of classicality.
 More specifically, in the subsection \ref{subsec:Physical}, we show that none of the layers of classicality respect transitivity property. Then we show that the concept like degree of broadcasting similar to degree of compatibility does not exist.  We also show the existence of informationally incomplete POVMs that are not individually broadcastable. Next we show that a set of broadcasting channels can be obtained through concatenation of broadcasting and non-disturbing channels. In the subsection \ref{subsec:Geometric},  we show that unlike compatibility, other layers of classicality are not convex, in general.  Finally in the subsection \ref{subsec:Relation}, We show  that specific type of concatenation relations among the broadcasting channels decide the layer in which a pair of observables resides.
 Finally, in section \ref{sec:conc} we summarize the work and discuss the future outlook.
\section{Preliminaries}\label{sec:prelimi}
In this section, we briefly discuss on compatibility, quantum channels, quantum instruments and the layers of classicality.
\subsection{Observables and Compatibility}
Let, $A$ be an observable acting on $d$ dimensional Hilbert space $\cH$ with number of outcomes $n_A$. Then we denote the outcome set of $A$ as $\Omega_A=\{1,....,n_A\}$. We denote the set of all observables  acting on $d$ dimensional Hilbert space $\cH$ with outcome set $X$ as $\cO^d_X$ and with outcome set $Y$ as $\cO^d_Y$. Then we denote the cartesian product of these two sets as $\cO^d_{XY}=\cO^d_X\times\cO^d_Y$. Then the set $\cO^d_{XY}$ is the set of all the pairs $(A,B)$ $\forall A\in \cO^d_X,\forall B\in\cO^d_Y$. From now on, through out the paper, we will consider all observables are acting on $d$ dimensional Hilbert space $\cH$ unless specifically mentioned.  A pair of observables $(A,B)\in \cO^d_{XY}$ is  compatible iff there exist an observable $G\in\cO^d_{Z}$ where $Z=X\times Y$, such that

\begin{eqnarray}
A(x)=\sum_y G(x,y)~\forall x\in X\\
B(y)=\sum_x G(x,y)~\forall y\in Y.
\end{eqnarray}

We denote the set of all compatible pairs acting on $d$ dimensional Hilbert space $\cH$ as $\cO^d_{comp,XY}$\cite{Heinosaari_Ziman}. Clearly, $\cO^d_{comp,XY}\subseteq \cO^d_{XY}$. Also $\cO^d_{comp,XY}$ is convex, where for $p\in[0,1]$, the convex combination of two observables $A\in\cO^d_X$ and $B\in\cO^d_X$ is defined as $pA+(1-p)B=\{pA(x)+(1-p)B(x)\}$ \cite{heino_incomp_wit}.\\
An observable $A=\{A(x)\}$ is called commutative if $A(x)A(y)=A(y)A(x)~\forall x,y\in\Omega_A$.
Throughout the paper, we restrict ourselves to the pairs of observables. A pair of observables $(A,B)\in\cO^d_{XY}$ is mutually commuting if $A(x)B(y)=B(y)A(x)$ for all $x\in X$ and $y\in Y$.

\subsection{Quantum Channels}
We denote the set of density matrices on a Hilbert space $\cH$ as $\cS(\cH)$. We also denote the set of linear operators on Hilbert space $\cH$ as $\cL(\cH)$.  A CPTP map $\Theta:\cS(\cH_{in})\rightarrow\cS(\cH_{out})$ is called a quantum channel, where $\cH_{in}$ and $\cH_{out}$ are two Hilbert spaces \cite{Heinosaari_Ziman}. The dual channel of $\Lambda$ is defined as a map $\Lambda^*:\cL(\cH_{out})\rightarrow\cL(\cH_{in})$ which satisfies the equation $\tr[\Lambda(\rho) A(x)]=\tr[\rho\Lambda^*(A(x))]$ for all $x\in\Omega_A$, for all states $\rho\in \cS(\cH_{in})$ and all observable $A=\{A(x)\}$ acting on $\cH_{out}$. We denote composition of two channels $\Lambda_1:\cS(\cH_{in})\rightarrow\cS(\cH^{\prime})$ and $\Lambda_2:\cS(\cH^{\prime})\rightarrow\cS(\cH_{out})$ as $\Lambda_2\circ\Lambda_1$, where $\cH^{\prime}$ is another Hilbert space. From the definition of dual channel we have
\begin{equation}
(\Lambda_2\circ\Lambda_1)^*=\Lambda_1^*\circ\Lambda_2^*.
\end{equation}
 We denote the set of all channels for which $\cH_{in}=\cH_{out}=\cH$ as $\cC^d$.\\

Our definition of broadcasting channel is same as the definition given in \cite{heino_layers_of_classicality}.

\begin{definition}
A channel $\Lambda :\cS(\cH)\rightarrow \cS(\cH_A\otimes \cH_B)$ with $\cH_A=\cH_B=\cH$ is called a broadcasting channel.
\end{definition}
We denote the set of all broadcasting channels acting on $\cS(\cH)$ as $\cC^d_{broad}$.
Let us now define a property of broadcasting channels which we named as \textit{local changeability.}

\begin{definition}
A broadcasting channel $\Lambda_1\in\cC^d_{broad}$ is called locally changeable to another broadcasting channel $\Lambda_2\in\cC^d_{broad}$ if there exist two channels $\Sigma_1\in\cC^d$ and $\Sigma_2\in\cC^d$ such that
\begin{equation}
\Lambda_2=(\Sigma_1\otimes\Sigma_2)\circ\Lambda_1
\end{equation}
and we denote it as $\Lambda_2\preceq_{local}\Lambda_1$. If both $\Lambda_2\preceq_{local}\Lambda_1$ and $\Lambda_1\preceq_{local}\Lambda_2$ hold then they are called locally interchangeable and we denote it as $\Lambda_1\simeq_{local}\Lambda_2$.
\end{definition}

We will use this property to construct relations between different layers of classicality.

\subsection{Quantum Instruments}
An  instrument $\cI$ with outcome set $X$ is the collection of CP maps $\cI=\{\cI_x:\cL(\cH)\rightarrow\cL(\cK)\}$ such that $\cI^{\cC}=\sum_x\cI_x$ is a quantum channel, where $\cH$ and $\cK$ are two Hilbert spaces. An observable $A\in\cO^d_X$ can be implemented using $\cI$ if for all $\rho\in\cS(\cH)$ and $x\in X$, $\tr[\rho A(x)]=\tr[\cI_x(\rho)]$. Such an instrument is known as $A$-compatible instrument and the channel $\cI^{\cC}$ is called compatible with $A$\cite{Heinosaari_Ziman}.
\subsection{Layers of Classicality}
In this subsection, we discuss briefly on the layers of classicality \cite{heino_layers_of_classicality}.
\begin{figure}[hbt!]
\includegraphics[width=7.5cm,height=4.3cm]{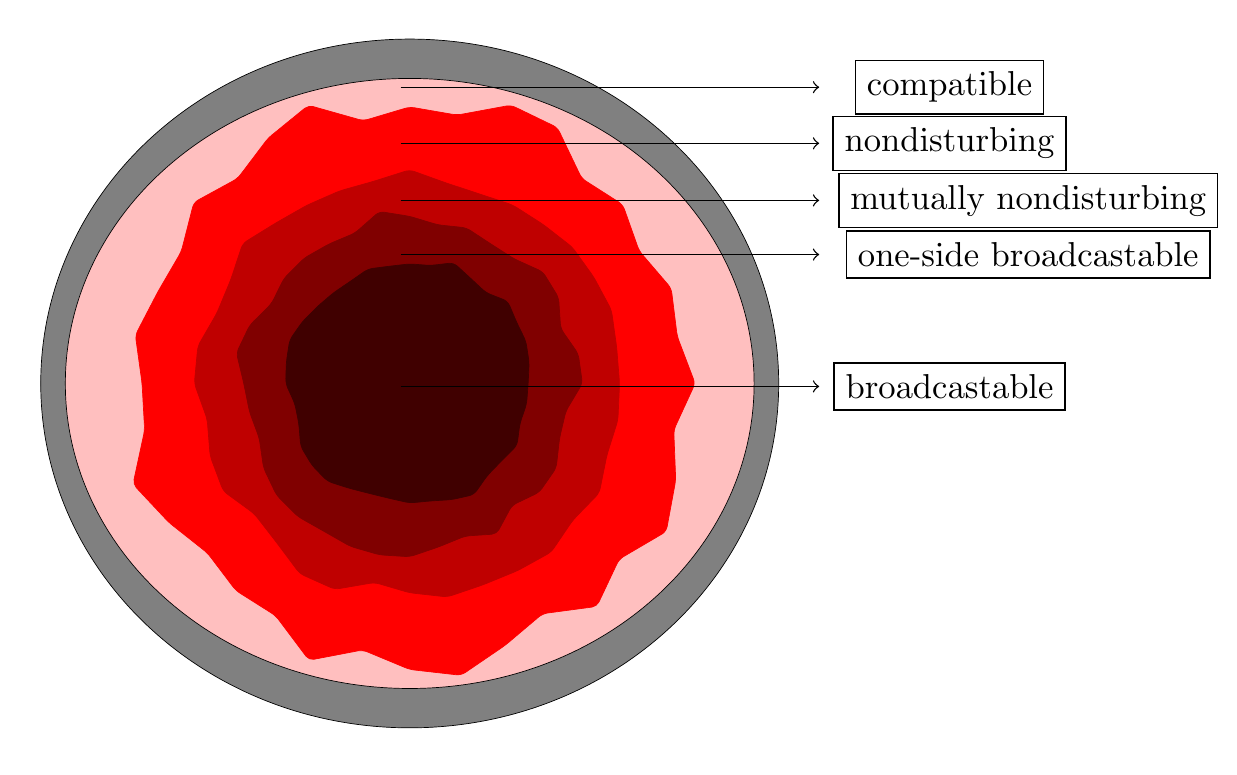}
\caption{(Colour online) Different layers of classicality.}\label{fig:layers}
\end{figure}
\subsubsection{Broadcastable pairs of observables}
\begin{definition}
An observable $A\in\cO^d_{X}$ is broadcast by a broadcasting channel $\Lambda\in\cC^d_{broad}$ if for all $\rho\in\cS(\cH)$, $x\in X$ 
\begin{align}
\tr[\rho A(x)]&=\tr[\Lambda(\rho) A(x)\otimes\Id]=\tr[\Lambda(\rho)\Id\otimes A(x)].
\end{align}
\end{definition}
A pair of observables $(A,B)\in\cO^d_{XY}$ is  broadcastable if there exists a broadcasting channel $\Lambda$ which can broadcast both observables $A$ and $B$. We denote the set of all broadcastable pairs of observables in $\cO^d_{XY}$ as $\cO^d_{broad,XY}$.\\
 It is proved in \citep{heino_layers_of_classicality} that a mutually commuting pair of commutative observables acting on finite dimensional Hilbert space is broadcastable. It is also proved in \cite{heino_layers_of_classicality} that a pair of qubit observables are broadcastable iff those observables are commutative and mutually commuting.

\subsubsection{One-side broadcastable pairs of observables}
\begin{definition}
A pair of observables $(A,B)\in\cO_{XY}$ is one-side broadcastable if there exists a broadcasting channel $\Lambda\in\cC^d_{broad}$ such that for all $\rho\in\cS(\cH)$, $x\in X$ and $y\in Y$
\begin{align}
&\tr[\rho A(x)]= \tr[\Lambda(\rho)A(x)\otimes\Id],\\
&\tr[\rho B(y)]= \tr[\Lambda(\rho)\Id\otimes B(y)]
\end{align}
\end{definition}
hold. We denote the set of all one side-broadcastable pairs of observables in $\cO^d_{XY}$ as $\cO^d_{1-side-broad,XY}$.

\subsubsection{Nondisturbing pairs of observables}
\begin{definition}
An observable $A\in\cO^d_X$ can be measured without disturbing $B\in\cO^d_Y$ if there exists an instrument $\cI=\{\cI(x)\}$ such that for all $\rho\in\cS(\cH)$, $x\in X$ and $y\in Y$
\begin{align}
&\tr[\rho A(x)]=\tr[\cI_x(\rho)],\\
&\tr[\rho B(y)]=\tr[\cI^{\cC}(\rho)B(y)].
\end{align}
\end{definition}
A pair of observables $(A,B)$ is called nondisturbing if either $A$ can be measured without disturbing $B$ or $B$ can be measured without disturbing $A$ \cite{heino_layers_of_classicality}.
We denote the set of all nondisturbing pairs of observables in $\cO^d_{XY}$ as $\cO^d_{nondist,XY}$. Some of the properties of nondisturbing pairs of observables has been discussed in the article  \cite{heino-wolf}. In the dimension more than two, the complete characterization of nondisturbing pairs is difficult and requires separate investigation.

\subsubsection{Mutually nondisturbing pairs of observables}
\begin{definition}
A pair of observables $(A,B)\in\cO_{XY}$ is mutually nondisturbing if $A$ can be measured without disturbing $B$ and $B$ can be measured without disturbing $A$.
\end{definition}
We denote the set of all mutually nondisturbing pairs of observables in $\cO^d_{XY}$ as $\cO^d_{mut-nondist,XY}$.\\
It is proved in \cite{heino_layers_of_classicality} that $\cO^d_{broad,XY}\subseteq\cO^d_{1-side-broad,XY}\subseteq\cO^d_{mut-nondist,XY}\subseteq \cO^d_{nondist,XY}\subseteq\cO^d_{comp,XY}$.\\
For any observable, if all of its POVM elements are multiple of identity, then the observable is called a trivial observable. Otherwise, it is called non-trivial observable. Any mutually commuting pair of non-trivial qubit observables is necessarily commutative \citep{heino_layers_of_classicality}. 
We now write down the following proposition which is originally proved in \citep{heino_layers_of_classicality}.
\begin{proposition}\label{prop:equi_state_heino}
For a pair of non-trivial qubit observables $A$ and $B$, the following statements are equivalent.
\begin{enumerate}
\item The pair $(A,B)$ is one-side broadcastable.
\item The pair $(A,B)$ is mutually nondisturbing.
\item The pair $(A,B)$ is nondisturbing.
\item The pair $(A,B)$ is mutually commuting.
\item The pair $(A,B)$ is broadcastable.
\end{enumerate}
\end{proposition}

\section{Main Results}\label{sec:main}
In this section, we discuss the main results of our work. At first, we discuss the physical properties of some layers of classicality. Then, we discuss geometric properties of different layers of classicality. Finally, we discuss the connection among different layers.
\subsection{Physical properties of some of the layers of classicality}\label{subsec:Physical}
In this section,  we discuss mainly the physical properties of broadcasting, one-side broadcasting and compatibility. Let us start with a simple property. We know that there exist three non-trivial observables $A,B$, and $C$ such that $(A,B)$ and $(B,C)$ are compatible pairs but the pair $(A,C)$ is not compatible i.e., compatibility does not respect transitivity property. Our next proposition shows that none of the layers of classicality respect this property.
\begin{proposition}
There exist non-trivial observables $A$, $B$, and $C$ such that the pairs $(A,B)$ and $(B,C)$ reside in a particular layer of classicality, while $(A,C)$ does not reside in that layer of classicality.\label{prop:transi}
\end{proposition}
\begin{proof}
Consider three non-trivial sharp observables acting on $\cH_2\otimes\cH_2$
\begin{eqnarray}
&A&=\{\ket{+0}\bra{+0},\ket{-0}\bra{-0}, \ket{01}\bra{01},\ket{11}\bra{11}\}\\
&B&=\{\Id\otimes\ket{0}\bra{0}, \ket{01}\bra{01},\ket{11}\bra{11}\}\\
&C&=\{\ket{00}\bra{00},\ket{10}\bra{10}, \ket{01}\bra{01},\ket{11}\bra{11}\}
\end{eqnarray}
where $\ket{0},\ket{1}$ are the eigenvectors of $\sigma_z$ and $\ket{+},\ket{-}$ are the eigenvectors of $\sigma_x$ and $\cH_2$ is Hilbert space with dimension $d=2$.\\
Clearly, these observables are commutative and also the pairs $(A,B)$ and $(B,C)$ are mutually commuting. Therefore, they are\\
broadcastable $\implies$ one side broadcastable $\implies$ mutually nondisturbing $\implies$ nondisturbing $\implies$ compatible.\\
But the pair $(A,C)$ is mutually noncommuting. Therefore, this pair is\\
incompatible $\implies$ disturbing $\implies$ mutually disturbing $\implies$ not one-side broadcastable $\implies$ not broadcastable.

So, the pairs $(A,B)$ and $(B,C)$ reside in all of these layers of classicality. But the pair $(A,C)$ does not reside in any of these layers of classicality.
\end{proof}

\begin{definition}
An observable $A_p$ is called an unsharp version of the observable $A$ with unsharp parameter $p$ iff $A_p=\{A_p(x)=pA(x)+(1-p)\frac{\Id}{n_A}\}$ for all $x\in\Omega_A$, where $n_A$ is the number of outcomes of $A$ and $p\in (0,1]$.
\end{definition}
Now we show a basic difference in a property between broadcasting and compatibility. Consider a pair of observables $(A,B)$ with number of outcomes $n_A$ and $n_B$. It is well known that for the particular value of $p,q\in(0,1]$, the compatibility of the pair $(A_p,B_q)$ does not imply compatibility of the pair $(A,B)$. But for broadcasting, this is not the case. The following theorem will clarify that.
\begin{theorem}
If the pair $(A_p,B_q)$ is broadcastable for some $p,q\in(0,1]$, then the pair $(A,B)$ is also broadcastable.\label{th:upsharp}
\end{theorem}

\begin{proof}
Suppose the pair $(A_p,B_q)$ is broadcastable by the broadcasting channel $\Lambda$ and  $p,q\in(0,1]$. Then for all $x\in\Omega_{A_p}$  and $\rho\in\cS(\cH)$ we have

\begin{align}
&\tr[\rho A_p(x)]=\tr[\Lambda(\rho)(A_p\otimes\Id)]\nonumber\\
or,~&\tr[\rho(pA(x)+(1-p)\frac{\Id}{n_A})]=\tr[\Lambda(\rho)((pA(x)\nonumber\\
&\hspace{3.8cm}+(1-p)\frac{\Id}{n_A})\otimes\Id)]\nonumber\\
or,~&p\tr[\rho(A(x)]+\frac{(1-p)}{n_A}=p\tr[\Lambda(\rho(A(x)\otimes\Id)]\nonumber\\
&\hspace{4.9cm}+\frac{(1-p)}{n_A}\nonumber\\
or,~&\tr[\rho(A(x)]=\tr[\Lambda(\rho(A(x)\otimes\Id)].
\end{align}

Similarly, one can prove for all $x\in\Omega_{A_p}$ $y\in\Omega_{B_q}$ and $\rho\in\cS(\cH)$
\begin{eqnarray}
\tr[\rho A(x)]&=\tr[\Lambda(\rho)(\Id\otimes A(x))],\\
\tr[\rho B(y)]&=\tr[\Lambda(\rho)(\Id\otimes B(y)],\\
\tr[\rho B(y)]&=\tr[\Lambda(\rho)(B(y)\otimes\Id)].\end{eqnarray}
Therefore, the pair $(A,B)$ is broadcastable by the same channel $\Lambda$.
\end{proof}

The highest value of a unsharp parameter for which the the unsharp version of a pair is compatible, is known as degree of compatibility\cite{heinosaari}. The Theorem \ref{th:upsharp} suggests that there is no concept like degree of broadcasting similar to degree of compatibility. In other words broadcastability of observables can not be quantified in this way!\\

One immediate corollary of  Theorem \ref{th:upsharp} is

\begin{corollary}
If a compatible pair $(A_p,B_q)$ is the unsharp version of an incompatible pair $(A,B)$ for some $p,q\in(0,1]$ then the pair $(A_p,B_q)$ is not broadcastable.\label{corollary__{incomp_broadcast}}
\end{corollary}

As we already know that there exist incompatible pairs of observables which has non-zero compatibility region, we conclude from the Corollary \ref{corollary__{incomp_broadcast}} that $\cO^d_{broad,XY}\subset \cO^d_{comp,XY}$.
Similarly, it is easy to prove that

\begin{corollary}
If the pair $(A,B)$ is broadcastable, then the pair $(A_p,B_q)$ is broadcastable for all $p,q\in[1,0]$.\label{coro:comp-broadcast}
\end{corollary}

Our, next proposition is very important and describes the relation between broadcasting of a compatible pair and broadcasting of it's joint observable. 

\begin{proposition}\label{proposi:joint_broad}
For any compatible pair $(A,B)\in\cO^d_{comp,XY}$, if the joint observable is broadcastable by a channel $\Lambda\in\cC^{d}_{broad}$, then the pair is also broadcastable by the same channel $\Lambda$.
\end{proposition}

\begin{proof}
Let the pair $(A,B)\in\cO^d_{XY}$ is compatible and $G\in{\cO_{X\times Y}}$ be the joint observable of it. Therefore for all $x\in X$ and $y\in Y$,

\begin{eqnarray}
A(x)=\sum_y G(x,y);
B(y)=\sum_x G(x,y).
\end{eqnarray}

Suppose that $G$ is broadcastable by $\Lambda\in\cC^d_{broad}$. 

Therefore,
\begin{align}
\tr[\rho A(x)]&=\tr[\rho \sum_y G(x,y)]\nonumber\\
&=\sum_y \tr[\rho G(x,y)]\nonumber\\
&=\sum_y \tr[\Lambda(\rho) G(x,y)\otimes\Id]\nonumber\\
&=\tr[\Lambda(\rho) \sum_y G(x,y)\otimes\Id]\nonumber\\
&=\tr[\Lambda(\rho)A(x) \otimes\Id]
\end{align}

Similarly, one can prove that

\begin{eqnarray}
\tr[\rho A(x)]&=\tr[\Lambda(\rho)\Id\otimes A(x)]\\
\tr[\rho B(y)]&=\tr[\Lambda(\rho)B(y) \otimes\Id]\\
\tr[\rho B(y)]&=\tr[\Lambda(\rho)\Id\otimes B(y)]
\end{eqnarray}

Hence, the proposition is proved.
\end{proof}

The immediate corollary of above proposition is

\begin{corollary}
If a compatible pair is not broadcastable, then their joint observable is also not broadcastable.\label{corollary:comp,broad}
\end{corollary}

We clarify this through our next example.

\begin{example}
Consider unsharp versions of a pair of spin-$\frac{1}{2}$ observables along $x$ and $y$ direction respectively with unsharpness parameter $\lambda=\frac{1}{\sqrt{2}}$. This pair is compatible\cite{Busch}. But from the Corollary \ref{corollary__{incomp_broadcast}}, we get that this pair is not broadcastable. Now, the joint observable of this pair is $G=\{\frac{1}{2}\ket{+\vec{n_1}}\bra{+\vec{n_1}}, \frac{1}{2}\ket{-\vec{n_1}}\bra{-\vec{n_1}}, \frac{1}{2}\ket{+\vec{n_2}}\bra{+\vec{n_2}},$ $\frac{1}{2}\ket{-\vec{n_2}}\bra{-\vec{n_2}}\}$, where, $\vec{n_1}=(1,1,0)$ and $\vec{n_2}=(1,-1,0)$. Therefore, from Corollary \ref{corollary:comp,broad} we get that this joint observable is not broadcastable. It should be noted that $G$ is not informationally complete.\label{example:no_brodcast}
\end{example}

Therefore, Corollary \ref{corollary__{incomp_broadcast}} and Corollary \ref{corollary:comp,broad} together  with example \ref{example:no_brodcast} enable us to write down the following theorem.

\setcounter{theorem}{1}
\begin{theorem}[No-broadcasting theorem for a single informationally incomplete POVM]
There exist some informationally incomplete POVMs which are not individually  broadcastable.\label{th:no_broad_single_POVM}
\end{theorem}

This theorem immediately raises a question-What is the minimum amount of extracted information which prohibits quantum broadcasting of observables? We keep this as open question.
Broadcasting is the weaker version of cloning, and therefore the Theorem \ref{th:no_broad_single_POVM} is a generalization of the no-cloning theorem for a single POVM given in \cite{Rastegin}.\\
Now, we define nondisturbing quantum channels with respect to a observable.
\begin{definition}
A channel $\Theta:\cS(\cH)\rightarrow\cS(\cH)$ is nondisturbing with respect to an observable $A$ if for all $\rho\in\cS(\cH)$,
\begin{eqnarray}
\tr[\rho A(x)]=\tr[\Theta(\rho)A(x)]~\forall~x\in\Omega_A
\end{eqnarray}
or, equivalently
\begin{equation}
\Theta^*(A(x))=A(x)~\forall~x\in\Omega_A.
\end{equation}

\end{definition}

A channel is nondisturbing for a set of observables iff it is nondisturbing for all observables in that set. Now, we prove our next proposition.

\begin{proposition}\label{proposi:broadcasting_preorder}
If the observable $A$ is broadcastable by the channel $\Lambda$ and if $\Theta:\cS(\cH)\rightarrow\cS(\cH)$ is a nondisturbing channel with respect to $A$ then $A$ is also broadcastable by $\Lambda\circ\Theta$.
\end{proposition}

\begin{proof}
The observable $A$ is broadcastable by $\Lambda$. Then for all $\rho\in\cS(\cH)$ and for all $x\in\Omega_A$

\begin{align}
\tr[(\Lambda\circ\Theta)(\rho)A(x)\otimes\Id]&=\tr[\Lambda(\Theta(\rho))A(x)\otimes\Id]\nonumber\\
&=\tr[\Theta(\rho)A(x)]\nonumber\\
&=\tr[\rho A(x)]
\end{align}

Similarly one can prove that

\begin{equation}
\tr[(\Lambda\circ\Theta)(\rho)\Id\otimes A(x)]=\tr[\rho A(x)].
\end{equation}
Therefore, $A$ is broadcastable by $\Lambda\circ\Theta$.

\end{proof}

For an observable, in this way, one can get the set of different broadcasting channels that broadcast the observable. It is an open question whether this set has the highest element i.e., a broadcasting channel in this set from which all other broadcasting channels that can broadcast the observable, can be constructed through concatenation with nondisturbing channels. It is also an open question whether this set has the lowest element i.e., a broadcasting channel in this set which can be constructed from all other broadcasting channels that broadcast the observable through concatenation with nondisturbing channels. \\
The results similar to the Theorem \ref{th:upsharp}, Corollaries \ref{corollary__{incomp_broadcast}}, \ref{coro:comp-broadcast}, \ref{corollary:comp,broad} and Propositions \ref{proposi:joint_broad}, \ref{proposi:broadcasting_preorder} can be similarly shown hold for one-side broadcastable pairs of observables.

\subsection{Geometric properties of different layers of classicality }\label{subsec:Geometric}
In this section, we discuss the geometric properties of the layers of classicality. Suppose $\chi_{XY}\subseteq\cO^d_{XY}$ is a set of pairs of observables with outcome set $X$ and $Y$ respectively. Then we denote the set of broadcasting channels which can broadcast all pairs of observables in the set $\chi$ as $\Gamma^d_{broad}(\chi_{XY})$. Similarly, for a set of broadcasting channels $\xi\subseteq \cC^d_{broad}$, we denote the set of all pairs of observables with outcome set $X$ and $Y$ respectively, which can be broadcast by all broadcasting channels in the set $\xi$ as $\cO^d_{broad,XY}(\xi)$. So, clearly $\cup_{\Lambda\in \cC^d_{broad}}\cO^d_{broad,XY}(\Lambda)=\cO^d_{broad,XY}$.

\begin{proposition}
For any pair of observables $(A,B)\in \cO^d_{XY}$, $\Gamma^d_{broad}(\{(A,B)\})$ is convex.\label{prop:convex_channel_broad}
\end{proposition}

\begin{proof}
Suppose both broadcasting channels $\Lambda_1$ and $\Lambda_2$ can broadcast the pair $(A,B)$. Then, for all  $p\in [0,1]$, $x\in X$ and $\rho\in(\cS(\cH))$

\begin{align}
&\tr[(p\Lambda_1+(1-p)\Lambda_2)(\rho) A(x)\otimes\Id]\nonumber\\
=~&p\tr[\Lambda_1(\rho) A(x)\otimes\Id]+(1-p)\tr[\Lambda_2(\rho) A(x)\otimes\Id]\nonumber\\
=~&p\tr[(\rho A(x)]+(1-p)\tr[\rho A(x)]\nonumber\\
=~&\tr[\rho A(x)]
\end{align}

In a similar way, one can easily prove that for all $p\in [0,1]$, $x\in X$, $y\in Y$ and $\rho\in(\cS(\cH))$
\begin{align}
\tr[(p\Lambda_1+(1-p)\Lambda_2)(\rho)\Id\otimes A(x)]=\tr[\rho A(x)]
\end{align}

and
\begin{align}
&\tr[(p\Lambda_1+(1-p)\Lambda_2)(\rho) B(y)\otimes\Id)]\nonumber\\
=~&\tr[(p\Lambda_1+(1-p)\Lambda_2)(\rho)\Id\otimes B(y)]\nonumber\\
=~&\tr[(\rho B(y)]
\end{align}

Therefore, the broadcasting channel $\Lambda^{\prime}=p\Lambda_1+(1-p)\Lambda_2$ can the pair $(A,B)$ for all $p\in [0,1]$.

\end{proof}

Similarly, we prove our next proposition.

\begin{proposition}
For any channel $\Lambda\in \cC^d_{broad}$, the set $\cO^d_{XY,broad}(\Lambda)$ is convex.\label{prop:convex_broad}
\end{proposition}

\begin{proof}
Suppose both the pairs $(A_1,B_1)\in\cO^d_{broad,XY}$  and $(A_2,B_2)\in\cO^d_{broad,XY}$ are broadcast by same channel $\Lambda$. Then for all $p\in [0,1]$, $x\in X$ and $\rho\in(\cS(\cH))$

\begin{align}
&\tr[\Lambda (\rho)(pA_1(x)+(1-p)A_2(x))\otimes\Id]\nonumber\\
=~&p\tr[\Lambda (\rho)A_1(x)\otimes\Id]+(1-p)\tr[\Lambda (\rho)A_2(x)\otimes\Id]\nonumber\\
=~&p\tr[\rho A_1(x)]+(1-p)\tr[(\rho)A_2(x)]\nonumber\\
=~&\tr[\rho(pA_1(x)+(1-p)A_2(x))].
\end{align}

In a similar way, one can easily prove that for all $p\in [0,1]$, $x\in X$, $y\in Y$ and $\rho\in(\cS(\cH))$ 

\begin{align}
&\tr[\Lambda (\rho)\Id\otimes(pA_1(x)+(1-p)A_2(x))]\nonumber\\
=~&\tr[\rho(pA_1(x)+(1-p)A_2(x))].
\end{align}

and

\begin{align}
&\tr[\rho(pB_1(y)+(1-p)B_2(y))]\nonumber\\
=~&\tr[\Lambda (\rho)\Id\otimes(pB_1(y)+(1-p)B_2(y))]\\
=~&\tr[\Lambda (\rho)(pB_1(y)+(1-p)B_2(y))\otimes\Id].
\end{align}

So, $\Lambda$ can broadcast the pair $(pA_1+(1-p)A_2,pB_1+(1-p)B_2)$ for all $p\in[0,1]$.
\end{proof}

Therefore, $\cO^d_{broad,XY}$ is union of these convex set. Now, suppose $(A_1,B_1)\in \cO_{broad,XY}$ and $(A_2,B_2)\in \cO_{broad,XY}$ are two pairs of observables. Let, $\lambda$ be an arbitrary number and $\lambda\in[0,1]$. Then consider another pair of observables $(A^{\lambda},B^{\lambda})$ such that $A^{\lambda}=\lambda A_1+(1-\lambda)A_2$ and $B^{\lambda}=\lambda B_1+(1-\lambda)B_2$. Then the following theorem holds.

\begin{theorem}
If two pairs $(A_1,B_1)\in \cO_{broad,XY}$ and $(A_2,B_2)\in \cO_{broad,XY}$ are not broadcastable by the same channel i.e., if $\Gamma^d_{broad}(\{(A_1,B_1),(A_2,B_2)\})=\emptyset$ where $\emptyset$ is a null set, then for all $p,q\in[0,1]$ and $p\neq q$, $$\Gamma^d_{broad}(\{(A^p,B^p),(A^q,B^q)\})=\emptyset.$$\label{th:noncovex_broad}
\end{theorem}

\begin{proof}
Suppose $(A^{\lambda},B^{\lambda})$ is broadcastable for all $\lambda\in [0,1]$ by the same broadcasting channel $\Lambda$. Then for all $\lambda\in [0,1]$, $x\in X$ and $\rho\in(\cS(\cH))$,

\begin{align}
&\tr[\rho (\lambda A_1(x)+(1-\lambda)A_2(x))]\nonumber\\
=~&\tr[\Lambda(\rho)(\lambda A_1(x)+(1-\lambda)A_2(x))\otimes\Id]\label{a1}\\
=~&\tr[\Lambda(\rho)\Id\otimes(\lambda A_1(x)+(1-\lambda)A_2(x))]\label{a2}
\end{align}

and

\begin{align}
&\tr[\rho (\lambda B_1(x)+(1-\lambda)B_2(x))]\nonumber\label{b1}\\
=~&\tr[\Lambda(\rho)(\lambda B_1(x)+(1-\lambda)B_2(x))\otimes\Id]\\
=~&\tr[\Lambda(\rho)\Id\otimes(\lambda B_1(x)+(1-\lambda)B_2(x))]\label{b2}
\end{align}

Now, let

\begin{eqnarray}
\tr(\rho A_1(x))-\tr[\Lambda(\rho)A_1(x)\otimes\Id]=p_1(x)\label{p1}\\
\tr(\rho A_2(x))-\tr[\Lambda(\rho)A_2(x)\otimes\Id]=q_1(x)\label{q1}\\
\tr(\rho B_1(y))-\tr[\Lambda(\rho)B_1(y)\otimes\Id]=r_1(y)\\
\tr(\rho B_2(y))-\tr[\Lambda(\rho)B_2(y)\otimes\Id]=s_1(y)
\end{eqnarray}

and

\begin{eqnarray}
\tr(\rho A_1(x))-\tr[\Lambda(\rho)\Id\otimes A_1(x)]=p_2(x)\\
\tr(\rho A_2(x))-\tr[\Lambda(\rho)\Id\otimes A_2(x)]=q_2(x)\\
\tr(\rho B_2(y))-\tr[\Lambda(\rho)\Id\otimes B_1(y)]=r_2(y)\\
\tr(\rho B_2(y))-\tr[\Lambda(\rho)\Id\otimes B_2(y)]=s_2(y)
\end{eqnarray}

Now it is to be noted that for all $i\in\{1,2\}$ and $x\in X$ and $y\in Y$, the numbers $p_1(x),q_i(x),r_i(y),s_i(y)$ are the differences of probabilities and all of these are independent of $\lambda$ and depend on $\Lambda$ and $\rho$. Therefore, moduli of them are less than or equals to 1. It is also to be noted that as $\Lambda$ can not broadcast individual pairs $(A_1,B_1)$ and $(A_2,B_2)$ simultaneously and therefore, for all $i\in\{1,2\}$ and $x\in X$ and $y\in Y$, the numbers $p_1(x),q_i(x),r_i(y),s_i(y)$ simultaneously can not be zero. Now suppose one of those numbers, say for example  $p_1(x)\neq 0$ for some $x\in X$. Then from equations \eqref{a1},\eqref{p1} and \eqref{q1} we have for all $\lambda\in[0,1]$

\begin{align}
&\lambda p_1(x)=(1-\lambda)(-q_1)\nonumber\\
& p_1(x)=\frac{(\lambda-1)}{\lambda}(q_1)\\
& \mid p_1(x)\mid=\mid\frac{(1-\lambda)}{\lambda}\mid\mid q_1 \mid.\label{fail_conv}
\end{align}

So, if $ q_1(x)=0$ above equation is satisfied only for $\lambda=0$. Now suppose $q_1(x)\neq 0$. Then, R.H.S of the equation\eqref{fail_conv} will vary with $\lambda$. But L.H.S.  of the equation\eqref{fail_conv} is independent of $\lambda$ and will not change. So, the equation\eqref{fail_conv} can be satisfied only for a particular value of $\lambda$ and can not be satisfied for all $\lambda\in[0,1]$.

Furthermore, note that for a particular channel of $\Lambda$, for $\lambda<\frac{\mid q_1(x)\mid}{1+\mid q_1(x)\mid}$, R.H.S of equation\eqref{fail_conv} is greater than 1. Therefore, since $\mid p_1(x)\mid\leq 1$, the equation \eqref{a1} can not be satisfied for all $\lambda\leq\frac{\mid q_1(x)\mid}{1+\mid q_1(x)\mid}$. Therefore, for any one or more than one of the numbers $p_1(x),q_i(x),r_i(y),s_i(y)$ for all $i\in\{1,2\}$ and $x\in X$ and $y\in Y$ to be greater than $0$, one can similarly prove that equations\eqref{a1},\eqref{a2}\eqref{b1} and \eqref{b2} can be satisfied simultaneously at most for a particular value of $\lambda \in[0,1]$. Therefore, a broadcasting channel $\Lambda$ can broadcast the pair $(A^{\lambda},B^{\lambda})$ at most for a particular value of $\lambda$. Hence, the theorem is proved.
\end{proof}

Our next proposition will clarify the geometric property of $\cO^d_{broad,XY}$ and the other layers.

\begin{proposition}
 $\cO^d_{broad,XY}$, $\cO^d_{1-side-broad,XY}$, $ \cO^d_{mut-nondist,X,Y}$ and $ \cO^d_{nondist,XY}$ are not convex, in general.\label{noncon1}
\end{proposition}

\begin{proof}
We prove this by showing a counter example. We know that for two bloch vectors $\vec{n_1}$ and $\vec{n_2}$
\begin{equation}
[\vec{n_1}.\vec{\sigma},\vec{n_2}.\vec{\sigma}]=2i(\vec{n_1}\times\vec{n_2}).\vec{\sigma}.\label{spin_com_rel}
\end{equation}
$\{\sigma_i|i=1,2,3\}$ are pauli matrices. We denote the eigen basis of $\sigma_z$ and $\sigma_x$ as $\{\ket{0},\ket{1}\}$ and $\{\ket{+},\ket{-}\}$ respectively. Now consider two observables $A_1$ and $A_2$ acting on two dimensional Hilbert space such that

\begin{align}
A_1(1)&=\frac{1}{2}\ket{0}\bra{0},\\
A_1(2)&=\frac{1}{2}\ket{1}\bra{1},\\
A_1(3)&=\frac{3}{8}\ket{0}\bra{0}+\frac{1}{8}\ket{1}\bra{1},\\
A_1(4)&=\frac{1}{8}\ket{0}\bra{0}+\frac{3}{8}\ket{1}\bra{1}
\end{align}

and

\begin{align}
A_2(1)&=\frac{1}{2}\ket{+}\bra{+},\\
A_2(2)&=\frac{1}{2}\ket{-}\bra{-},\\
A_2(3)&=\frac{1}{8}\ket{+}\bra{+}+\frac{3}{8}\ket{-}\bra{-},\\
A_2(4)&=\frac{3}{8}\ket{+}\bra{+}+\frac{1}{8}\ket{-}\bra{-}.
\end{align}

Again, consider two sharp spin $\frac{1}{2}$ observables $B_1$ and $B_2$ acting on two dimensional Hilbert space such that

\begin{align}
&B_1(1)=\ket{0}\bra{0};\hspace{0.3cm}B_1(2)=\ket{1}\bra{1}, \\
&B_2(1)=\ket{+}\bra{+};\hspace{0.1cm}B_2(2)=\ket{-}\bra{-}.
\end{align}

Clearly, the both of the pairs $(A_1,B_1)$ and $(A_2,B_2)$ are commutative and mutually commuting and therefore, both of the pairs are individually broadcastable.

Now, let $A^{\prime}=\frac{1}{2}A_1+\frac{1}{2}A_2$ and $B^{\prime}=\frac{1}{2}B_1+\frac{1}{2}B_2$. We know that a pair of qubit observables are broadcastable iff those observables are commutative and mutually commuting\cite{heino_layers_of_classicality}. If we show that $A^{\prime}$ is not commutative then the pair $(A^{\prime},B^{\prime})$ is not broadcastable .

Now, it is to be noted that

\begin{align}
A^{\prime}(1)&=\frac{1}{2}\left[\frac{1}{2}\ket{0}\bra{0}+\frac{1}{2}\ket{+}\bra{+}\right]\nonumber\\
&=\frac{1}{8}\left[2.\Id+(\sigma_z+\sigma_x)\right].
\end{align}

Similarly,
\begin{align}
A^{\prime}(3)&=\frac{1}{2}\left[\frac{3}{8}\ket{0}\bra{0}+\frac{1}{8}\ket{1}\bra{1}+\frac{1}{8}\ket{+}\bra{+}+\frac{3}{8}\ket{-}\bra{-}\right]\nonumber\\
&=\frac{1}{16}\left[4.\Id+(\sigma_z-\sigma_x)\right].
\end{align}

Therefore, using equation \eqref{spin_com_rel} we get $[A^{\prime}(1),A^{\prime}(3)]\neq 0$. Hence, we have proved that  $\cO^2_{broad,X_1Y_1}$ is not convex, where $X_1=\{1,2,3,4\}$ and $Y_1=\{1,2\}$. Several other counter examples exist to prove this proposition.
Now as the pairs $(A_1,B_1)$, $(A_2,B_2)$ and $(A^{\prime},B^{\prime})$ are the pairs of non-trivial qubit observables, from the Proposition \ref{prop:equi_state_heino}, we get that  $\cO^2_{1-side-broad,X_1Y_1}$, $\cO^2_{mut-nondist,X_1,Y_1}$ and $ \cO^2_{nondist,X_1Y_1}$ are also not convex.
This is one of the essential differences between compatibility and any  other layer of classicality in geometric properties.
\end{proof}

   There may be some other differences in geometric properties. But it needs further investigation to find out that.\\
 \textit{From the proof of propositions \ref{prop:convex_channel_broad},\ref{prop:convex_broad} and theorem\ref{th:noncovex_broad} it is clear that similar results hold for one-side broadcasting.}\\
 Below in the Table \ref{tab:properties}, the comparison of several properties of the different layers of classicality has been presented.

  \begin{table}[h]  \centering
\begin{tabular}{ |p{2.4cm}|p{1.8cm}|p{1.8cm}|p{1.8cm}|} 
 \hline
 Name of the layers & Transitivity property & Convexity property & Degree of the layer\\
 \hline
 \hline
 Broadcastable pairs & $\times$ (Prop. \ref{prop:transi}) & $\times$ (Prop. \ref{noncon1}) &$\times$ (Th.\ref{th:upsharp})\\
 \hline
 One-side Broadcastable pairs &$\times$ (Prop. \ref{prop:transi})& $\times$ (Prop. \ref{noncon1})&$\times$ (similar to Th.\ref{th:upsharp})\\ 
 \hline
Mutually non-disturbing pairs &$\times$ (Prop. \ref{prop:transi})& $\times$ (Prop. \ref{noncon1})&?\\ 
\hline
Non-disturbing pairs &$\times$ (Prop. \ref{prop:transi})& $\times$ (Prop. \ref{noncon1})& ? \\
\hline
Compatible pairs &$\times$ (Prop. \ref{prop:transi})& \checkmark (Ref.\cite{heino_incomp_wit})&\checkmark (Ref.\cite{heinosaari})\\
 \hline
\end{tabular}
 \label{tab:properties} 
  \caption{
  Comparison of transitivity property, convexity property and existence of the concept of degree of the layers (i.e., the highest value of unsharp parameter for which the unsharp version of a pair reside in a particular layer) is presented in this table. (\checkmark) indicates that a particular property is satisfied by a layer or a concept exist for that layer, ($\times$) indicates the opposite to that of (\checkmark) and (?) indicates that it is not known whether a particular property is satisfied by a layer or a concept exist for that layer.}
  \end{table}  

\subsection{Relations among the layers of classicality}\label{subsec:Relation}
Before investigating relations among different layers of classicality let us define some strict non-overlapping layers of classicality.

\begin{definition}
Let $(A,B)\in \cO_{XY}$ be a pair of observables. Then
\begin{enumerate}
\item The pair is called weakly compatible if the pair is compatible, but not nondisturbing.
\item The pair is called weakly nondisturbing if the pair is nondisturbing, but not mutually nondisturbing.
\item The pair is called weakly mutually nondisturbing if the pair is mutually nondisturbing, but not one-side broadcastable.
\item The pair is called weakly one-side broadcastable  if the pair is one-side broadcastable, but not broadcastable.
\end{enumerate}\label{def:st_layers}
\end{definition}

This layers in Definition \ref{def:st_layers} are the strips between two successive layers of classicality and can be easily visualised from the figure \ref{fig:layers}.
Let us now investigate the relation between broadcasting and one-side broadcasting. In the following theorems and proposition we will see that specific concatination relations between the broadcasting channels decide the resident of the pair in a particular layer.  We start with our next theorem.\\

\begin{theorem}[Broadcasting and one-side broadcasting]
Consider a pair $(A,B)\in\cO^d_{XY}$ and also suppose that $A$ is broadcastable by $\Lambda_1$ and $B$ is broadcastable by $\Lambda_2$. Then the pair is atleast one-side broadcastable if there exist two channels $\Sigma,\Theta\in \cC^d$ such that $\Lambda_2=(\Theta\otimes\Sigma)\circ\Lambda_1$, where $\Theta\in \cC^d$ is a nondisturbing channel with respect to $A$ and $\Sigma\in \cC^d$ is an arbitrary channel. Along with this, if the pair is also not broadcastable, then it is weakly one-side broadcastable.
\end{theorem}

\begin{proof}
 $A$ is broadcast by $\Lambda_1$.  We know that dual channels are unital. Then for all  $x\in X$, and $\rho\in(\cS(\cH))$ we have
\begin{align}
\tr[\Lambda_2(\rho)A(x)\otimes\Id]&=\tr[(\Theta\otimes\Sigma)\Lambda_1(\rho)A(x)\otimes\Id]\nonumber\\
&=\tr[\Lambda_1(\rho)(\Theta\otimes\Sigma)^*(A(x)\otimes\Id)]\nonumber\\
&=\tr[\Lambda_1(\rho)(\Theta^*(A(x))\otimes\Sigma^*(\Id))]\nonumber\\
&=\tr[\Lambda_1(\rho)A(x)\otimes\Id]\nonumber\\
&=\tr(\rho A(x))\label{rel_brod_1_side_broad_1}
\end{align}

Now  $B$ is broadcast by $\Lambda_2$. Then for $y\in Y$ and $\rho\in\cS(\cH)$,
\begin{eqnarray}
\tr(\rho B(y))=\tr[\Lambda_2(\rho)\Id\otimes B(y)]\label{rel_brod_1_side_broad_2}
\end{eqnarray}
From equations \eqref{rel_brod_1_side_broad_1} and equation \eqref{rel_brod_1_side_broad_2}, we get that the pair $(A,B)$ is one-side broadcastable.
\end{proof}

Before, we investigate the relation between one-side broadcasting and nondisturbance, let us write down the proposition which is originally proved in \cite{heino_layers_of_classicality}.

\begin{proposition}
An observable $A\in\cO^d_{X}$ can be measured without disturbing an observable $B\in{\cO^d_Y}$ iff there exists a  $d^{\prime}$ dimensional ancilla Hilbert space $\cK$, a channel $\Lambda:\cS(\cH)\rightarrow \cS(\cK\otimes \cH)$ and an observable $A^{\prime}\in\cO^{d^{\prime}}_{X}$ acting on the ancilla Hilbert space $\cK$ such that for any state $\rho\in\cS(\cH)$, $x\in X$ and $y\in Y$, we have

\begin{eqnarray}
\tr[\rho(A(x))]=\tr[\Lambda(\rho)A^{\prime}(x)\otimes \Id]\\
\tr[\rho(B(y))]=\tr[\Lambda(\rho)\Id\otimes B(y)].
\end{eqnarray}\label{proposi:Heino}
\end{proposition}

Now let us investigate the connection between  broadcasting and nondisturbance. Our next theorem states one of the possible relations.

\begin{theorem}[Broadcasting and nondisturbance]
If $A\in\cO^d_X$ is broadcastable by the broadcasting channel $\Lambda_1\in\cC^d_{broad}$ and $B\in\cO^d_Y$ is broadcastable by the broadcasting channel $\Lambda_2\in\cC^d_{broad}$, then $A$ can be measured without disturbing $B$ if $\Lambda_2$ is locally changeable to $\Lambda_1$. Along with this, if the pair is also not mutually nondisturbing, then it is weakly nondisturbing.\label{th:broad_nondist}
\end{theorem}

\begin{proof}
If  $\Lambda_2$ is locally changeable to $\Lambda_1$ then there exist two channels $\Sigma_1,\Sigma_2\in\cC^d$ such that

\begin{equation}
\Lambda_1=(\Sigma_1\otimes\Sigma_2)\circ\Lambda_2.
\end{equation}

Since, $A$ is broadcastable by $\Lambda_1$, for all $\rho\in\cS(\cH)$ and for all $x\in X$
\begin{align}
\tr[\rho A(x)]&=\tr[\Lambda_1(\rho) A(x)\otimes\Id ]\nonumber\\
&=\tr[(\Sigma_1\otimes\Sigma_2)\circ\Lambda_2(\rho) A(x)\otimes\Id ]\nonumber\\
&=\tr[\Lambda_2(\rho) \Sigma_1^*(A(x))\otimes\Sigma_2^*(\Id) ]\nonumber\\
&=\tr[\Lambda_2(\rho) \Sigma_1^*(A(x))\otimes\Id)]\label{broadcast_nondist_1}
\end{align}
Since $B$ is broadcast by $\Lambda_2$, all $\rho\in\cS(\cH)$ and for all $y\in Y$

\begin{align}
\tr[\rho B(y)]=\tr[\Lambda_2(\rho)\Id\otimes B(y)]\label{broadcast_nondist_2}.
\end{align}

Choosing $A^{\prime}=\{A^{\prime}(x)\}=\{\Sigma_1^*(A(x))\}$ and from equation \eqref{broadcast_nondist_1}, equation \eqref{broadcast_nondist_2}, and Proposition \ref{proposi:Heino} we get that the pair $(A,B)$ is a nondisturbing pair. Hence, the theorem is proved.
\end{proof}

From Theorem \ref{th:broad_nondist} we immediately get our next proposition.

\begin{proposition}[Broadcasting and mutual nondisturbance]
If $A\in\cO^d_X$ is broadcastable by the broadcasting channel $\Lambda_1\in\cC^d_{broad}$ and $B\in\cO^d_Y$ is broadcastable by the broadcasting channel $\Lambda_2\in\cC^d_{broad}$, then $A$ and $B$ are at least mutually nondisturbing if $\Lambda_2$ and $\Lambda_1$ are locally interchangeable. Along with this, if the pair is also not one-side broadcastable, then it is weakly mutually nondisturbing.\label{th:broad_mut_nondist}
\end{proposition}
The results similar to Theorem \ref{th:broad_nondist} and Proposition \ref{th:broad_mut_nondist} hold if $A$ and $B$ are one-side broadcastable by $\Lambda_1$ and $\Lambda_2$ respectively to two different sides.\\
We obtained connections between different layers of classicality. To find out more connections similar to the connections derived in this subsection, it needs further investigation.

\section{Conclusion}\label{sec:conc}
We have studied the physical and geometric properties of different layers of classicality. We have  shown that such properties of different layers of classicality have differences as well as similarities. In particular, We have shown that: (i) none of the layers of classicality respect transitivity property, (ii) the concept like degree of broadcasting similar to degree of compatibility does not exist, (iii) there exist informationally incomplete POVMs that are not individually broadcastable, (iv) a set of broadcasting channels can be obtained through concatenation of broadcasting and non-disturbing channels, (v) unlike compatibility, other layers of classicality are not convex, in general. Finally, we have obtained connections among different layers of classicality. All of our results are valid for all finite dimensions except Propostion \ref{noncon1}. It is not known whether Proposition \ref{noncon1} is valid for the dimensions more than two. \\
This work opens up several new avenues of research on the compatibility of POVMs and the other layers of classicality. Firstly, it is not known whether the set of all broadcasting channels that broadcast a  particular observable has the greatest or the lowest element. Secondly, we do not know which layers are open or closed among the layers of classicality, except for compatibility. Thirdly, we do not know whether the set of all nondisturbing pairs of observables and the set of all mutually nondisturbing pairs of observables are non-convex in higher dimensions. It will be also interesting to construct wittnesses and resource theory for the different layers. Fourthly, till now we do not have a full set of connections among different layers of classicality and also do not have the full set of mathematical properties. Fifthly, it is possible to generalise atleast some of our results for a set of $n$ observables with $n>2$.\\
An important avenue for future research is to find out how some of these layers of classicality provide advantages over other layers of classicality in several information-theoretic tasks.\\
\section{Acknowledgements}
I would like to thank my advisor Prof. Sibasish Ghosh and my co-advisor Prof. Prabha Mandayam for their support and valuable comments on this work.

\bibliography{layers_of_classicality_ref_v3}

\end{document}